\documentclass[12pt,a4paper]{article} 
\usepackage[margin=1in]{geometry}
\usepackage[T1]{fontenc}
\usepackage{amsmath,dsfont}
\usepackage{caption}

\usepackage[utf8]{inputenc}
\usepackage{bm}
\usepackage{verbatim}
\usepackage{float}
\usepackage{mwe}
\usepackage{color,soul}
\usepackage{threeparttable}
\usepackage{graphicx}
\usepackage{subcaption}
\usepackage{mwe}
\usepackage{mdframed}
\usepackage{xcolor}

\usepackage{natbib}
\usepackage{multirow}
\usepackage{makecell}
\usepackage{booktabs}
\usepackage{array}
\usepackage{hhline}

\usepackage{mathtools}
\usepackage{bbm}
\usepackage{romanbar}
\usepackage{amsthm}
\usepackage{enumerate}
\usepackage{rotating}
\usepackage{tikz}
\usetikzlibrary{matrix}
\usetikzlibrary{calc,intersections}
\usepackage{mathrsfs}
\usepackage[ruled,vlined]{algorithm2e}

 \usepackage{setspace}
\usepackage{color}

\newtheorem{theorem}{Theorem}[section]

\newtheorem{proposition}[theorem]{Proposition}

\theoremstyle{definition}

\theoremstyle{remark}

\numberwithin{equation}{section}

\newcommand{\bol}[1]{\mbox{\boldmath$#1$}}

\newcommand{\eqdist}{\stackrel{d}{=}}

\newcommand{\bmu}{\bol{\mu}}

\newcommand{\btheta}{\bol{\theta}}

\newcommand{\bx}{\mathbf{x}}

\newcommand{\bw}{\mathbf{w}}

\newcommand{\bS}{\mathbf{S}}

\providecommand{\keywords}[1]
{
\small	
\textbf{\textit{Keywords}:} #1
}

\usepackage{authblk}
\title{Volatility Sensitive Bayesian Estimation of Portfolio VaR and CVaR}
\author[1]{Taras Bodnar}
\author[1]{Vilhelm Niklasson\footnote{Corresponding author: Vilhelm Niklasson. E-mail address: niklasson@math.su.se.}}
\author[1]{Erik Thorsén}
\affil[1]{\small\textit{Department of Mathematics, Stockholm University, SE-10691 Stockholm, Sweden}}
\date{}

\begin{document}
\maketitle

\begin{abstract}
In this paper, a new way to integrate volatility information for estimating value at risk (VaR) and conditional value at risk (CVaR) of a portfolio is suggested. The new method is developed from the perspective of Bayesian statistics and it is based on the idea of volatility clustering. By specifying the hyperparameters in a conjugate prior based on two different rolling window sizes, it is possible to quickly adapt to changes in volatility and automatically specify the degree of certainty in the prior. This constitutes an advantage in comparison to existing Bayesian methods that are less sensitive to such changes in volatilities and also usually lack standardized ways of expressing the degree of belief.  We illustrate our new approach using both simulated and empirical data. Compared to some other well known homoscedastic and heteroscedastic models, the new method provides a good alternative for risk estimation, especially during turbulent periods where it can quickly adapt to changing market conditions.
\end{abstract}


\section{Introduction}\label{intro}
Bayesian statistics was introduced in portfolio analysis in the 1970s with the pioneering works of \citet[]{winkler1973bayesian}, \citet{Barry1974}, and \citet{KleinBawa1976}. One of the main advantages of the Bayesian approach compared to traditional plug-in methods is that the Bayesian methodology makes it possible to employ useful prior information. How to include such prior information has been a topic of interest during the last decades and several different approaches have been suggested. A well known method was proposed by \citet{black1992global} who provided an informal Bayesian model based on market equilibrium arguments. Another popular approach was suggested by \citet{pastor2000portfolio} and \citet{pastor2000comparing} in which asset pricing models were used to include prior information. \citet{tu2010incorporating} also suggested a Bayesian model where priors were imposed on the solution of a problem rather than on the original model parameters. Although these methods provide good ways to specify the prior belief, they do not enable a systematic quantification of the certainty in that belief. This puts the investor in a difficult position since it is generally not obvious how certain he or she should be in the prior. Moreover, these methods focus mainly on finding optimal portfolio weights and it is not easy to calibrate them for different purposes, such as risk estimation. 

Since more than a decade, the Basel regulations require banks to use value at risk (VaR) and conditional value at risk (CVaR) in their risk assessments \citep{basel4}. All trading desks that are bounded by these regulations must perform regular backtesting of their risk models. The backtesting procedure required by Basel is based on a binomial test of a 'hit sequence' where each number in the sequence is 1 if the daily VaR has been exceeded and 0 if it has not been exceeded. Hence it resembles the testing procedure suggested by \citet{christoffersen1998evaluating}. Based on the outcome of the test, different capital requirements are imposed on the bank and models that perform badly will even be prohibited. Many alternative ways of performing VaR and CVaR backtesting have recently been suggested in the literature (see, e.g., \citet{escanciano2012pitfalls}, \citet{ziggel2014new}, \citet{kratz2018multinomial}, and references therein). However, no matter the backtesting procedure, the estimation of VaR and CVaR remains a challenging task, especially, when considering a portfolio consisting of several assets since the number of parameters to estimate grows very fast. 

An important factor which determines the forecast quality is estimation error and its impact on the forecasting procedure.
Recently, \citet{kerkhof2004backtesting} and \citet{du2017backtesting} investigated what effect the estimation uncertainty has on exceedances of VaR in the case of a single asset. Both papers derived the variance of the asymptotic distributions of the suggested test statistics. They deduced that the effect can be large, depending on the ratio of the out-of-sample and in-sample period. However, the application of Bayesian methods completely avoids this issue, since it automatically takes the estimation error into account.

Another reason why risk determination is a challenging task is because the asset returns are not homogeneous. A well-known characteristic of asset returns is that large or small deviations tend to group together in time. This can clearly be seen in financial data from, for instance, the financial crises in 2008 or, more recently, the Covid-19 outbreak in 2020. This phenomenon is known as volatility clustering and it has been documented in numerous articles \citep[see, e.g.,][]{ding1993long, ding1996modeling}. Different models have been suggested to cope with this, most notably the GARCH model \citep{engle1982autoregressive}. However, due to the dimensionality issue of the full multivariate GARCH model, the less general DCC-GARCH model is commonly used to model multiple financial assets \citep{engle2002dynamic}. The model has also been applied to estimate VaR and CVaR of portfolios \citep[see, e.g.,][]{lee2006study, santos2013comparing}. However, fitting a DCC-GARCH model is still quite computationally demanding since it usually involves maximizing a potentially complicated likelihood function. Moreover, there is no guarantee that a global optimum is reached.  

This paper contributes to the current literature on Bayesian portfolio analysis and risk estimation by suggesting a new and simple algorithm for specifying the hyperparameters in a conjugate prior which takes volatility clustering into account. The new method also enables automatic specification of the degree of belief in the prior, hence it resolves the current issue of uncertainty specification. Simulated and empirical data from S\&P 500 are used to illustrate how the suggested method performs and how it compares to other well-known methods with respect to Basel backtesting, especially during the recent Covid-19 outbreak.  
 
The paper is organized as follows. In Section 2, we present a Bayesian model of portfolio returns and we provide expressions of VaR and CVaR derived by using the posterior predictive distribution. Section 3 presents the new algorithm for specifying the hyperparameters in the conjugate prior. In Section 4, the Basel procedure for backtesting VaR and CVaR is described. Section 5 contains a simulation study where the new method is compared to other existing approaches. The comparison is continued in Section 6 with real market data. Finally, Section 7 contains the conclusions.
\section{A Bayesian model of portfolio returns and estimation of VaR and CVaR}\label{sec:model}
Let $\bx_t$ be the $k$-dimensional vector of asset returns at time $t$ and let $\bx_{(t-1)}=(\bx_{t-n},...,\bx_{t-1})$ stand for the observation matrix of the asset returns $\bx_{t-n}$, ..., $\bx_{t-1}$ taken from time $t-n$ until $t-1$. The vector of portfolio weights is denoted by $\bw$ and it determines how much the investor owns of each asset. We assume that the whole investor's wealth is shared between the selected assets, i.e., $\mathbf{1}_k^\top \bw=1$ where $\mathbf{1}_k$ denotes the $k$-dimensional vector of ones. The return of the portfolio with weights $\bw$ at time point $t$ is then given by
\begin{equation}\label{eq:portfolio_return}
    X_{P,t} = \bw^\top\bx_t.
\end{equation}

From the Bayesian perspective, the posterior predictive distribution of $X_{P,t}$, i.e., the conditional distribution of $X_{P,t}$ given $\bx_{(t-1)}$, is computed by \citep[see, e.g., p. 244 in][]{bernardo2000bayesian}
\begin{equation}
    f(x_{P,t}|\bx_{(t-1)}) = \int_{\theta \in \Theta} f(x_{P,t}|\btheta)\pi(\btheta|\bx_{(t-1)})d\btheta,
    \label{eq:posterior}
\end{equation}
where $\Theta$ denotes the parameter space, $f(\cdot|\btheta)$ is the conditional density of $X_{P,t}$ given $\btheta$, and $\pi(\btheta|\bx_{(t-1)})$ stands for the posterior distribution of $\btheta$ given $\bx_{(t-1)}$. The posterior predictive distribution possesses a number of applications, especially in the prediction of the future realization of the portfolio return. In fact, any Bayesian point estimate, like the posterior predictive mean, posterior predictive mode, or posterior predictive median, reflects the investor expectation about the future portfolio return.

Assume that the asset returns $\bx_i, i \in \{t-n,...,t\}$, are multivariate normally distributed and independent when conditioning on mean vector $\bm{\mu}$ and covariance matrix $\bm{\Sigma}$. The conjugate prior is then given by
\begin{equation}\label{eq:conjugate_prior}
    \bmu|\bm{\Sigma} \sim N_k\left(\mathbf{m}_0,\frac{1}{r_0}\bm{\Sigma}\right)
\quad \text{and} \quad
    \bm{\Sigma} \sim IW_k (d_0, \bS_0),
\end{equation}
where $\bm{m}_0$, $r_0$, $\bS_0$ and $d_0$ are hyperparameters, $N_k\left(\mathbf{m}_0,\frac{1}{r_0}\bm{\Sigma}\right)$ denotes the multivariate normal distribution with mean vector $\mathbf{m}_0$ and covariance matrix $\frac{1}{r_0}\bm{\Sigma}$, and $IW_k(d_0, \bS_0)$ denotes the inverse Wishart distribution with $d_0$ degrees of freedom and parameter matrix $\bS_0$.

It can be shown that the posterior predictive distribution \eqref{eq:posterior} under these assumptions is a generalized Student's $t$-distribution when using the conjugate prior \citep[see, e.g.,][]{winkler1973bayesian, BodnarLindholmNiklassonThorsen2020}. More precisely, let $\widehat{X}_{P,t}$ be a random variable which follows the posterior predictive distribution, then
\begin{equation}\label{eq:stochastic_rep_common}
\widehat{X}_{P,t} \eqdist \mathbf{w}^T\bar{\bx}_{t-1}+\tau(d_{k,n})\sqrt{r_{k,n}}\sqrt{\mathbf{w}^T\bS_{t-1}\mathbf{w}},
\end{equation}
where the right hand side is implicitly conditioned on observed data, $\tau(d_{k,n})$ denotes a random variable following the standard $t$-distribution with $d_{k,n} = n+d_0-2k$ degrees of freedom and
\begin{equation}\label{eq:conjugate_r_kn_and_x_t_minus_1}
r_{k,n}=\frac{n+r_0+1}{(n+r_0)(n+d_0-2k)}, \quad 
\bar{\bx}_{t-1} = \frac{n\bar{\bm{x}}+r_0\mathbf{m_0}}{n+r_0},
\end{equation}
\begin{equation}\label{eq:conjugate_S_t_minus_1}
\bS_{t-1} = \sum_{i=t-n}^{t-1}(\bx_i-\bar{\bm{x}})(\bx_i-\bar{\bm{x}})^T+\bS_0+nr_0\frac{(\mathbf{m_0}-\bar{\bx}_{t-1})(\mathbf{m_0}-\bar{\bx}_{t-1})^\top}{n+r_0},
\end{equation}
where $\bar{\bm{x}}=\frac{1}{n}\sum_{i=t-n}^{t-1}\bx_i$ denotes the sample mean.

Using \eqref{eq:stochastic_rep_common}, it follows directly by the definitions of VaR and CVaR at level $\alpha \in (0.5,1)$ that the portfolio risk using these risk measures (commonly denoted by $Q$) is given by
\begin{equation}\label{eq:VaR_CVaR_common}
    Q_{t-1}(\bw) =-\mathbf{w}^T\bar{\bx}_{t-1} +q_{\alpha}\sqrt{r_{k,n}}\sqrt{\mathbf{w}^T\bS_{t-1}\mathbf{w}},
\end{equation}
where $q_{\alpha}$ depends on if VaR or CVaR is considered. More precisely, let $d_{\alpha}(d_{k,n})$  be the $\alpha$ quantile of the $t$-distribution with $d_{k,n}$ degrees of freedom. Then 
\begin{equation}\label{eq:q_alpha_VaR}
    q_{\alpha} = d_{\alpha}(d_{k,n})
\end{equation}
when considering VaR, and 
\begin{equation}\label{eq:q_alpha_CVaR}
 q_{\alpha} = \frac{1}{1-\alpha}\frac{\Gamma\left(\frac{d_{k,n}+1}{2}\right)}
    {\Gamma\left(\frac{d_{k,n}}{2}\right)
    \sqrt{\pi d_{k,n}}}
    \frac{d_{k,n}}{d_{k,n}-1}
    \left(1+\frac{(d_{\alpha}(d_{k,n}))^2}{d_{k,n}}\right)^{-\frac{d_{k,n}-1}{2}}
\end{equation}
when considering CVaR \citep[see, ][for details]{BodnarLindholmNiklassonThorsen2020}.
\section{Volatility sensitive conjugate hyperparameters}
The conjugate prior defined in \eqref{eq:conjugate_prior} is an informative prior which has been used extensively in the Bayesian financial literature and many different ways have been suggested on how to specify its four hyperparameters. Some of the methods use only historical data of the assets returns following the empirical Bayes approach \citep[see, e.g.,][] {frost1986empirical, kolm2017bayesian, bauder2018bayesian} whereas others use additional input, for instance, from asset pricing models \citep[see, e.g.,][] {pastor2000portfolio, pastor2000comparing}.

Using only historical data where all vectors of the asset returns are of equal importance, as in the empirical Bayes approach, works well when the market conditions are stable. However, it is well known that markets behave irregularly and volatility clustering is often observed in financial data. Existing asset pricing models usually focus on the mean behaviour of asset returns, while a little attention is paid to capture volatility clustering \citep[see, e.g,][]{pastor2000portfolio}.
Manually expressing the hyperparameters based on a personal belief is also very difficult. For moderate portfolio sizes the covariance matrix is hard to specify, as it includes a positive definite constraint which is highly nonlinear. Moreover, it is very difficult to specify the degree of belief in the covariance matrix. 

We suggest a new automatic way to specify the hyperparameters which deals with all of the above issues. In order to capture volatility clustering it is of particular importance to specify the hyperparameters $d_0$ and $\bm{S}_0$ since they determine the prior distribution of $\bm{\Sigma}$. The idea behind the suggested approach is to make the specification of these parameters based on a comparison of the long and short term behaviour of the portfolio variance. Let $n$ be the period corresponding to the long term and let $n_r\le n$ correspond to the most recent period. The new method of specifying $d_0$ and $\bm{S}_0$ can then be defined as in Algorithm \ref{alg:conjugate_hyperparameters}.
 
\begin{algorithm}
\caption{Volatility sensitive conjugate hyperparameters}
\label{alg:conjugate_hyperparameters}
\SetAlgoLined
    \begin{enumerate}
    \item Calculate the sample estimates $\bm{\sigma}$ of the stock standard deviations using the recent $n$ observation vectors of the asset returns.\\
	 \item Calculate the sample estimates $\bm{\sigma}_r$ of the stock standard deviations using the recent $n_r$ observation vectors of the asset returns together with the long term sample mean. \\
	\item Define $\bm{D}$  to be a diagonal matrix with diagonal elements given by $\bm{\sigma}_r/\bm{\sigma}$, where the division is element-wise. \\
	\item Calculate the sample covariance matrix $\hat{\bm{\Sigma}}$ using the recent $n$ observation vectors of the asset returns. \\
	\item Define $\hat{\bm{\Sigma}}_r = \bm{D}\hat{\bm{\Sigma}}\bm{D}$ to be the sample covariance matrix with variances based on the recent period. \\
	\item Define $V_{r,\bm{w}}=\bm{w}^{\top}\hat{\bm{\Sigma}}_r\bm{w}$ and $V_{\bm{w}}=\bm{w}^{\top}\hat{\bm{\Sigma}}\bm{w}$ to be the estimated portfolio variances based on the short and long terms, respectively.
	\item Set $d_0 = \max\left(k+2, n \left(\max\left(1, \frac{V_{r,\bm{w}}}{V_{\bm{w}}}\right)\right)^h \left(\max\left(1, \frac{V_{\bm{w}}}{V_{r,\bm{w}}}\right)\right)^l\right)$.\\
	\item Set $\bm{S}_0 = \frac{(d_0-k-1)(n-1)}{n} \hat{\bm{\Sigma}}_{r}$.
	\end{enumerate}
\end{algorithm}

Specifying $\bm{S}_0 $ according to Algorithm \ref{alg:conjugate_hyperparameters} means that our prior belief about $\bm{\Sigma}$ is that the variances will be similar to what has been observed recently whereas the prior correlations are based on the long term behaviour of the stocks. Moreover, our degree of belief, which is controlled by $d_0$, is determined by how much the portfolio variance differs between the long and short term periods. Hence more pronounced deviations from historical data could motivate a greater or smaller belief in the recent observations depending on the signs and magnitudes of $l$ and $h$.
 
To better understand the implication of the algorithm, note that the second summand of VaR and CVaR estimates given by \eqref{eq:VaR_CVaR_common} can be rewritten as
\begin{align}
    q_{\alpha} \sqrt{r_{k,n}} \sqrt{\bm{w}^\top \bm{S}_{t-1} \bm{w}} & =  q_{\alpha} \sqrt{r_{k,n} (n-1)} \sqrt{\frac{1}{n-1}\bm{w}^\top \bm{S}_{t-1} \bw} \\
    &=q_{\alpha} \sqrt{\frac{(n+r_0+1)(n-1)}{(n+r_0)(n+d_0-2k)}} \sqrt{\frac{1}{n-1}\bm{w}^\top \bm{S}_{t-1} \bw}.\label{eq:portfolio_quantile}
\end{align}
For large values of $n$ and if $V_{r,\bm{w}}> V_{\bm{w}}$, we have that $d_0=n(V_{r,\bm{w}}/V_{\bm{w}})^h$ and
\begin{equation}
\sqrt{\frac{(n+r_0+1)(n-1)}{(n+r_0)(n+d_0-2k)}} \approx \sqrt{\frac{1}{1+(V_{r,\bm{w}}/V_{\bm{w}})^h}}.
\end{equation}
Moreover, if $\bm{m}_0$ is set to be the sample mean, from \eqref{eq:portfolio_quantile} and \eqref{eq:conjugate_S_t_minus_1} we get
\begin{equation}
\frac{1}{n-1}\bm{w}^\top \bm{S}_{t-1} \bm{w} = V_{\bm{w}} + \frac{d_0-k-1}{n} V_{r,\bm{w}}  \approx V_{\bm{w}} + (V_{r,\bm{w}}/V_{\bm{w}})^h V_{r,\bm{w}}.
\end{equation}
Thus, \eqref{eq:portfolio_quantile} becomes approximately equal to
\begin{equation}\label{eq:convex_arrrox_rep}
q_\alpha \sqrt{\frac{1}{1+(V_{r,\bm{w}}/V_{\bm{w}})^h}V_{\bm{w}} + \frac{(V_{r,\bm{w}}/V_{\bm{w}})^h}{1+(V_{r,\bm{w}}/V_{\bm{w}})^h} V_{r,\bm{w}}},
\end{equation}
which can be simplified to 
\begin{equation}
q_\alpha \sqrt{V_\bw}\sqrt{\frac{1+(V_{r,\bm{w}}/V_{\bm{w}})^{h+1}}{1+(V_{r,\bm{w}}/V_{\bm{w}})^h}}.
\end{equation}
The quantity inside the square root of \eqref{eq:convex_arrrox_rep} is a convex combination between the variances based on the recent period and the long term period where the weights between the periods are determined by $(V_{r,\bm{w}}/V_{\bm{w}})$ and $h$. An analogous reasoning can be made when $V_{r,\bm{w}} < V_{\bm{w}}$.

Algorithm \ref{alg:conjugate_hyperparameters} has several interesting properties. Some of these properties concern comparisons to the empirical Bayes estimate of $\bm{S}_0$ which is given by \citep[see, e.g.,][]{bauder2020bayesian}
\begin{equation}\label{eq:empirical_bayes_S_0}
\bm{S}_0=\frac{(d_0-k-1)(n-1)}{n}\hat{\bm{\Sigma}}.
\end{equation}
We summarize the properties in Proposition \ref{prop:properties} below. Following the empirical Bayes approach we set $\bm{m}_0=\bar{\bx}$.
\begin{proposition}\label{prop:properties} 
Let $\bS_{0}$ and $d_0$ be determined by Algorithm \ref{alg:conjugate_hyperparameters}. Then, it holds that
\begin{enumerate}
    \item VaR and CVaR estimates of a portfolio using the new method are monotonically increasing with respect to $h$ if $n\rightarrow \infty$.
    \item VaR and CVaR estimates of a portfolio using the new method are monotonically decreasing with respect to $l$ if $n\rightarrow \infty$.
    \item The new method is equivalent to empirical Bayes if $n_r = n$, $n\ge k+2$ and $d_0=n$ is used in empirical Bayes.
    \item VaR and CVaR estimates of a portfolio using the new method are always greater than corresponding empirical Bayes estimates when $V_{r,\bm{w}}>V_{\bm{w}}$, $d_0=n$ is used in empirical Bayes and $n \rightarrow \infty$.
    \item VaR and CVaR estimates of a portfolio using the new method are always smaller than corresponding empirical Bayes estimates when $V_{r,\bm{w}}< V_{\bm{w}}$, $d_0=n$ is used in empirical Bayes and $n \rightarrow \infty$.
\end{enumerate}
\end{proposition}
\begin{proof} $ $\newline
\begin{enumerate}
\item Only the case $V_{r,\bm{w}}>V_{\bm{w}}$ has to be considered since $h$ will not matter otherwise. Note that \eqref{eq:convex_arrrox_rep} will become the limit of \eqref{eq:portfolio_quantile} when $n\rightarrow \infty$ and $q_{\alpha}$ will be given by a normal distribution. Increasing $h$ in this limit means that the factor in front of $V_{r,\bm{w}}$ will become larger in the convex combination between $V_{\bm{w}}$ and $V_{r,\bm{w}}$. Hence the result follows since $V_{r,\bm{w}}>V_{\bm{w}}$, and VaR and CVaR estimates given by \eqref{eq:VaR_CVaR_common} are increasing with respect to the second summand.  

\item This follows from similar arguments as Property 1.

\item Setting $n_r=n$ in Algorithm \ref{alg:conjugate_hyperparameters} means that $\bm{D}$ will be the identity matrix and thus $d_0$ will be equal to $n$ in the new method. Furthermore, $\bm{S}_0$ given by the new method will coincide with the empirical Bayes estimate of $\bm{S}_0$ in \eqref{eq:empirical_bayes_S_0}. Hence the identity holds. 

\item Using empirical Bayes, $\bS_0$ is given by \eqref{eq:empirical_bayes_S_0} together with $d_0 = n$. It holds that \eqref{eq:portfolio_quantile} converges to $q_{\alpha} \sqrt{V_{\bm{w}}}$ as $n \rightarrow \infty$,
where $q_{\alpha}$ corresponds to a quantile of the standard normal distribution. The quantity $q_{\alpha} \sqrt{V_{\bm{w}}}$ is smaller than the corresponding quantity \eqref{eq:convex_arrrox_rep} when using the new method whenever $V_{r,\bm{w}} > V_{\bm{w}}$. Thus the result follows since VaR and CVaR estimates given by \eqref{eq:VaR_CVaR_common} are increasing with respect to the second summand. 
\item This follows from similar arguments as Property 4.
\end{enumerate}
\end{proof}

The new method of setting the hyperparameters in the conjugate prior requires the specification of three additional parameters: $n_r$, $h$ and $l$. The parameter $n_r$ controls how long the short term period should be whereas the parameters $h$ and $l$ determine how much recent high or low portfolio variances (compared to the long run) should impact our certainty. These hyperparameters can be specified using some loss function together with historical data or they can be specified based on personal belief. Compared to the original hyperparameters, these new parameters are likely easier to interpret, especially for someone who is not very familiar with Bayesian statistics.

Finally, it should be mentioned that the hyperparameters $r_0$ and $\bm{m}_0$ are of less importance for capturing volatility clustering and thus are not included in Algorithm \ref{alg:conjugate_hyperparameters}. The prior belief about the mean, $\bm{m}_0$, can be specified according to the empirical Bayes method as suggested before Proposition \ref{prop:properties}. It is also possible to specify $\bm{m}_0$ based on some market equilibrium arguments or capital asset pricing models \citep[see, e.g.,][]{black1992global, pastor2000comparing}. The hyperparameter $r_0$ controls the uncertainty of the mean and can be set to some constant based on personal belief.

\section{Basel backtesting of VaR and CVaR}
Basel requires banks to report CVaR at the 97.5 \% level and do backtesting of it using VaR estimates at levels 97.5 \% and 99 \% \citep{basel4}. The test is commonly referred to as the 'traffic light test' since the outcome of the test is classified as different colors. Basically, the test is a standard binomial test based on a hit sequence,  i.e., a series indicating exceedance or non-exceedance of VaR for each trading day. Let $\{I_t\}_{0\le t \le T}$ denote such a hit sequence between times $0$ and $T$, i.e., each element $I_t$ is 
\begin{equation}
\begin{cases}
    1 \text{ if the portfolio return at time } t \text{ is smaller than the estimated -VaR,}\\
    0 \text{ otherwise.}
\end{cases}
\end{equation}
By definition, $1-\alpha$ of the portfolio returns should be smaller than -VaR at level $\alpha$ when VaR is correctly determined. The test hypotheses expressing the probability of 'success' in the binomial distribution are
\begin{equation}\label{hyp-VaR}
H_0:\, p_I= 1-\alpha
\quad \text{against} \quad
H_1:\, p_I\neq 1-\alpha .\end{equation}

Basel requires banks to calculate the number of VaR exceedances during one year (about 250 trading days), i.e., the number of cases when the method of VaR estimation fails to provide a good forecast. Let $C_{T}=\sum_{t=1}^T I_t$ with usually $T=250$ and let $c$ denote the observed value of $C_T$ computed by using the data at hand. Then the cumulative probability of the number of exceedances, i.e., the number of method failures, is computed under the null hypothesis and it is given by
\[P_T=I\!\!P(C_T \le c)
\quad \text{with} \quad
C_T \sim \text{Binomial}(T,1-\alpha).\]
If the probability for the number of exceedances 
$P_T$ is smaller than 95 \%, then the result is classified as 'Green'. It is classified as 'Amber' if the probability $P_T$ is between 95 \% and 99.99 \%. Otherwise, if the probability $P_T$ is greater than 99.99 \%, the result is classified as 'Red'\footnote{Basel only explicitly specify the colors for VaR backtests at the 99 \% level but we have generalized it also to the 97.5 \% level since this level is also required when doing backtesing.}. The number of exceedances and the classification zones will impact the capital requirements and possibly make a model invalid to use.

\section{Simulation study}
We now illustrate how the new approach performs with respect to VaR estimation when using a portfolio consisting of several assets. Three different simulation scenarios are considered and the proposed approach is compared to three other VaR estimation methods using the Basel backtest.

\subsection{Simulation setup}
Different simulation scenarios for the stock returns are considered in order to see how the new approach performs under different market assumptions. The three scenarios are referred to as:
\begin{enumerate}
    \item Multivariate normal distribution (MVN)
    \item Perturbed multivariate normal distribution (PMVN)
    \item Multivariate GARCH (MGARCH)
\end{enumerate}
In the MVN scenario, the simulated returns follow a multivariate normal distribution with fixed mean vector and covariance matrix. 

The PMVN scenario is similar to the MVN scenario with the exception that the standard deviations corresponding to the multivariate normal distribution potentially deviate from the basis values during certain periods of time. The first step when implementing this scenario is to randomly select a period length. The length of a period is either three, four or five days with equal probabilities. Next, the period is classified as 'low volatility', 'normal volatility' or 'high volatility' with probabilities 0.05, 0.9 and 0.05, respectively. If the period is classified as a 'low volatility' period, the basis standard deviations are multiplied by uniformly distributed random variables that take values between 0.5 and 0.7. Similarly, if it is a 'high volatility' period, the standard deviations are scaled up by factors that are uniformly distributed between 1.5 and 3. During 'normal volatility' periods the standard deviations are not scaled at all. 

The MGARCH method uses the DCC-GARCH(1,1) method with normally distributed residuals to generate vectors of the asset returns.

The parameters of the models in the different simulation scenarios are determined by the corresponding estimators obtained by fitting the model to real market data. We use randomly selected market data from stocks in the S\&P 500 index between 2018 and 2020. Once the parameters have been estimated, 500 returns are generated using the simulation scenario in consideration of which the last 250 are used to evaluate the models VaR estimations according to Basel with a rolling window size of 250 days. In addition to the new approach based on the volatility sensitive conjugate method, the VaR models in consideration are the empirical Bayes conjugate method (referred to as EB) with $d_0=r_0=n$ \citep[see,][]{BodnarLindholmNiklassonThorsen2020}, the method based on the samples estimates \citep[see, e.g,][]{alexander2002economic} and the DCC-GARCH(1,1) method \citep[see, e.g.,][]{lee2006study}. 

Two different parameter combinations of $n_r$, $h$ and $l$ are considered for the new volatility sensitive conjugate method. In the first combination we use $n_r = 4$, $h=2$ and $l=0$ and it is referred to as VS(4,2,0). This should reflect the parameters of an investor who is risk averse since the risk is scaled up quickly when high volatilities are observed but it is not scaled down so much when the recent volatilities are low. In the second combination we use  $n_r = 4$, $h=0$ and $l=0$ and it is referred to as VS(4,0,0). This should reflect the setting of an investor who has a fixed degree of belief in the prior distribution of the covariance matrix. In both cases we use the empirical Bayes approach to specify $\bm{m}_0$ and we set $r_0$ equal to $n$.

Equally weighted portfolios of three different sizes are considered, namely 5, 10 and 15. For each portfolio size and simulation scenario we generate 100 series of returns. This is to make the comparison more rigorous and avoid the common mistake of sherry picking stocks or parameters. Hence in total we get 100 outcomes from the Basel backtest for each simulation scenario and each VaR estimation model.

\subsection{Simulation results}
Table \ref{tab:sim_normal} shows the proportion of VaR estimations classified as 'Green', 'Amber' and 'Red' in the Basel backtest using the different estimation methods when returns are simulated using MVN. It is observed in Table \ref{tab:sim_normal} that all risk estimation methods perform quite similarly when the returns follow a multivariate normal distribution. However, a close look reveals that the estimation methods based on the volatility sensitive conjugate method and the empirical Bayes conjugate method might slightly overestimate the risk whereas the methods based on the sample estimates and DCC-GARCH(1,1) might slightly underestimate the risk.  

\begin{singlespace}
\begin{table}[H]
\footnotesize
\begin{mdframed}[backgroundcolor=black!10,rightline=false,leftline=false]
\centering
\caption{Proportion of VaR estimations classified as 'Green', 'Amber' and 'Red' in the Basel backtest using simulated multivariate normal returns.  }
\label{tab:sim_normal}
{\begin{tabular}{| c | c | c | c c c c c |}
  \hline
  \makecell{VaR \\level} & \makecell{Portfolio\\ size} & \makecell{Basel \\ zone} & \makecell{VS\\(4,2,0)} & \makecell{VS\\(4,0,0)} & EB & Sample & \makecell{DCC-GARCH\\(1,1)}\\
  \hhline{|=|=|=|=====|}
    \multirow{9}{*}{97.5 \%} & \multirow{4}{*}{5} & Green & 1.00  & 0.99  & 0.97 & 0.96 & 0.94 \\
    & & Amber & 0.00 & 0.01  & 0.03 & 0.04 & 0.06 \\
    & & Red & 0.00 & 0.00  & 0.00 & 0.00 & 0.00 \\
    \cline{2-8}
    & \multirow{4}{*}{10} & Green & 1.00 & 0.99 & 0.97 & 0.97 & 0.92 \\
    & & Amber & 0.00 & 0.01 & 0.03 & 0.03 & 0.08\\
    & & Red & 0.00 & 0.00 & 0.00 & 0.00 & 0.00 \\
    \cline{2-8}
    & \multirow{4}{*}{15} & Green & 1.00 & 1.00 & 0.98 & 0.97 & 0.95 \\
    & & Amber & 0.00 & 0.01 & 0.02 & 0.03 & 0.05\\
    & & Red & 0.00 & 0.00 & 0.00 & 0.00 & 0.00\\
    \cline{1-8}
    \multirow{9}{*}{99 \%} & \multirow{4}{*}{5} & Green & 0.95  & 0.92  & 0.93 & 0.90 & 0.92\\
    & & Amber & 0.05 & 0.08  & 0.07 & 0.10 & 0.08 \\
    & & Red & 0.00 & 0.00  & 0.00 & 0.00  & 0.00 \\
    \cline{2-8}
    & \multirow{4}{*}{10} & Green & 0.96 & 0.95 & 0.90 & 0.89 & 0.87\\
    & & Amber & 0.04 & 0.05 & 0.10 & 0.11 & 0.13 \\
    & & Red & 0.00 & 0.00 & 0.00 & 0.00 & 0.00\\
    \cline{2-8}
    & \multirow{4}{*}{15} & Green & 0.93 & 0.93 & 0.88 & 0.85 & 0.88\\
    & & Amber & 0.07 & 0.07 & 0.12 & 0.15 & 0.12 \\
    & & Red & 0.00 & 0.00 & 0.00 & 0.00 & 0.00\\
    \cline{1-8}
\end{tabular}}{}
\end{mdframed}
\end{table}
\end{singlespace}

Similar results but using the perturbed multivariate normal distribution to generate the returns are shown in Table \ref{tab:sim_pnormal}. We observe that the different methods perform quite similarly for the lower VaR level when using perturbed multivariate normal returns, although the new volatility sensitive conjugate method might be a bit too risk conservative. For the higher VaR level we note that the new method has a clear advantage. This is because this method can easily capture rapid changes in the market conditions corresponding to the 'high volatility' periods. Surprisingly, the DCC-GARCH(1,1) which is also a heteroscedastic model, does not seem to capture such changes better than the method based on the sample estimates.

\begin{singlespace}
\begin{table}[H]
\footnotesize
\begin{mdframed}[backgroundcolor=black!10,rightline=false,leftline=false]
\centering
\caption{Proportion of VaR estimations classified as 'Green', 'Amber' and 'Red' in the Basel backtest using simulated perturbed multivariate normal returns. }
\label{tab:sim_pnormal}
{\begin{tabular}{| c | c | c | c c c c c |}
  \hline
  \makecell{VaR \\level} & \makecell{Portfolio\\ size} & \makecell{Basel \\ zone} & \makecell{VS\\(4,2,0)} & \makecell{VS\\(4,0,0)} & EB & Sample & \makecell{DCC-GARCH\\(1,1)}\\
  \hhline{|=|=|=|=====|}
    \multirow{9}{*}{97.5 \%} & \multirow{4}{*}{5} & Green & 1.00  & 1.00  & 0.97 & 0.97 & 0.94 \\
    & & Amber & 0.00 & 0.00  & 0.03 & 0.03 & 0.06 \\
    & & Red & 0.00 & 0.00  & 0.00 & 0.00 & 0.00 \\
    \cline{2-8}
    & \multirow{4}{*}{10} & Green & 1.00 & 0.99 & 0.97 & 0.96 & 0.95 \\
    & & Amber & 0.00 & 0.01 & 0.03 & 0.04 & 0.05\\
    & & Red & 0.00 & 0.00 & 0.00 & 0.00 & 0.00 \\
    \cline{2-8}
    & \multirow{4}{*}{15} & Green & 0.99 & 0.98 & 0.96 & 0.94 & 0.92 \\
    & & Amber & 0.01 & 0.02 & 0.04 & 0.05 & 0.08\\
    & & Red & 0.00 & 0.00 & 0.00 & 0.01 & 0.00\\
    \cline{1-8}
    \multirow{9}{*}{99 \%} & \multirow{4}{*}{5} & Green & 0.97  & 0.95  & 0.90 & 0.90 & 0.74\\
    & & Amber & 0.03 & 0.05  & 0.10 & 0.10 & 0.26 \\
    & & Red & 0.00 & 0.00  & 0.00 & 0.00  & 0.00 \\
    \cline{2-8}
    & \multirow{4}{*}{10} & Green & 0.97 & 0.93 & 0.81 & 0.80 & 0.82\\
    & & Amber & 0.03 & 0.07 & 0.19 & 0.20 & 0.17 \\
    & & Red & 0.00 & 0.00 & 0.01 & 0.00 & 0.01\\
    \cline{2-8}
    & \multirow{4}{*}{15} & Green & 0.94 & 0.90 & 0.81 & 0.79 & 0.78\\
    & & Amber & 0.06 & 0.10 & 0.19 & 0.21 & 0.22 \\
    & & Red & 0.00 & 0.00 & 0.00 & 0.00 & 0.00\\
    \cline{1-8}
\end{tabular}}{}
\end{mdframed}
\end{table}
\end{singlespace}

The simulation results using the MGARCH scenario are presented in Table \ref{tab:sim_dcc}. We see that the method based on the volatility sensitive conjugate prior seems to slightly overestimate the risk whereas the other methods seem to underestimate the risk. It is interesting to note that even the DCC-GARCH(1,1) method, which in this case corresponds to the true underlying process, does not perform very well when the portfolio size is large in relation to the sample size. The latter result can be explained by the increasing impact of the estimation error on the VaR forecast when the portfolio size, i.e., the number of model parameters, becomes large.

\begin{singlespace}
\begin{table}[H]
\footnotesize
\begin{mdframed}[backgroundcolor=black!10,rightline=false,leftline=false]
\centering
\caption{Proportion of VaR estimations classified as 'Green', 'Amber' and 'Red' in the Basel backtest using simulated MGARCH returns. }
\label{tab:sim_dcc}
{\begin{tabular}{| c | c | c | c c c c c |}
  \hline
  \makecell{VaR \\level} & \makecell{Portfolio\\ size} & \makecell{Basel \\ zone} & \makecell{VS\\(4,2,0)} & \makecell{VS\\(4,0,0)} & EB & Sample & \makecell{DCC-GARCH\\(1,1)}\\
  \hhline{|=|=|=|=====|}
    \multirow{9}{*}{97.5 \%} & \multirow{4}{*}{5} & Green & 1.00  & 0.95  & 0.89 & 0.89 & 0.97 \\
    & & Amber & 0.00 & 0.05  & 0.08 & 0.08 & 0.03 \\
    & & Red & 0.00 & 0.00  & 0.03 & 0.03 & 0.00 \\
    \cline{2-8}
    & \multirow{4}{*}{10} & Green & 0.99 & 0.98 & 0.92 & 0.91 & 0.97 \\
    & & Amber & 0.01 & 0.02 & 0.05 & 0.06 & 0.03\\
    & & Red & 0.00 & 0.00 & 0.03 & 0.03 & 0.00 \\
    \cline{2-8}
    & \multirow{4}{*}{15} & Green & 0.98 & 0.97 & 0.89 & 0.85 & 0.91\\
    & & Amber & 0.02 & 0.03 & 0.09 & 0.12 & 0.09 \\
    & & Red & 0.00 & 0.00 & 0.02 & 0.03 & 0.00\\
    \cline{1-8}
    \multirow{9}{*}{99 \%} & \multirow{4}{*}{5} & Green & 1.00  & 0.97  & 0.75 & 0.73 & 0.91\\
    & & Amber & 0.00 & 0.03  & 0.19 & 0.21 & 0.09 \\
    & & Red & 0.00 & 0.00  & 0.06 & 0.06  & 0.00 \\
    \cline{2-8}
    & \multirow{4}{*}{10} & Green & 0.95 & 0.93 & 0.82 & 0.80 & 0.88\\
    & & Amber & 0.05 & 0.07 & 0.14 & 0.16 & 0.12 \\
    & & Red & 0.00 & 0.00 & 0.04 & 0.04 & 0.00\\
    \cline{2-8}
    & \multirow{4}{*}{15} & Green & 0.96 & 0.91 & 0.75 & 0.74 & 0.86\\
    & & Amber & 0.04 & 0.09 & 0.22 & 0.23 & 0.14 \\
    & & Red & 0.00 & 0.00 & 0.03 & 0.03 & 0.00\\
    \cline{1-8}
\end{tabular}}{}
\end{mdframed}
\end{table}
\end{singlespace}

\section{Empirical analysis}
We now present how the new volatility sensitive conjugate method performs compared to other methods when using real market data from 2019 and 2020.

\subsection{Setup of empirical study}
In the empirical comparison we use daily stock returns from randomly selected stocks in the S\&P 500 index. As in the simulation study, we consider portfolios of three different sizes, namely 5, 10 and 15. The assets in the portfolios are selected at random to avoid picking assets that are not representative for the rest of the market. For each portfolio size, 100 groups of assets are selected from S\&P 500 at the beginning of either 2019 or 2020, and 100 equally weighted portfolios are constructed that remain unchanged throughout the year. For each portfolio, we estimate VaR on a daily basis using a rolling window size of $n=250$ and we consider the same estimation methods as in the simulation study. At the end of the year we calculate the number of VaR exceedances at the 97.5 \% and 99 \% levels and classify the result according to the Basel backtest. 

The years 2019 and 2020 were selected for the empirical comparison since they are two very recent years. Moreover, the former can be considered a typical stable year, whereas the latter includes a very turbulent period in March due to the Covid-19 outbreak. Hence they should theoretically be representative for different kinds of market conditions.

\subsection{Empirical results}
Table \ref{tab:empirical_2019} shows the proportion of portfolios classified in the 'Green', 'Amber' and 'Red' zones using the different VaR estimation methods in 2019. As can be seen in Table \ref{tab:empirical_2019}, all of the methods perform quite similarly during 2019 with respect to the Basel backtest. All of them show good results at the $97.5 \%$ VaR level whereas the result looks slightly worse at the $99 \%$ level. Especially the DCC-GARCH(1,1) method seems to underestimate the risk at this level during this year by having many portfolios in the 'Amber' zone. 

\begin{singlespace}
\begin{table}[H]
\footnotesize
\begin{mdframed}[backgroundcolor=black!10,rightline=false,leftline=false]
\centering
\caption{Proportion of VaR estimations classified as 'Green', 'Amber' and 'Red' in the Basel backtest using different portfolio sizes and VaR estimation methods in 2019. }
\label{tab:empirical_2019}
{\begin{tabular}{| c | c | c | c c c c c |}
  \hline
  \makecell{VaR \\level} & \makecell{Portfolio\\ size} & \makecell{Basel \\ zone} & \makecell{VS\\(4,2,0)} & \makecell{VS\\(4,0,0)} & EB & Sample & \makecell{DCC-GARCH\\(1,1)}\\
  \hhline{|=|=|=|=====|}
    \multirow{9}{*}{97.5 \%} & \multirow{4}{*}{5} & Green & 1.00  & 0.99  & 1.00 & 1.00 & 0.97 \\
    & & Amber & 0.00 & 0.01  & 0.00 & 0.00 & 0.03 \\
    & & Red & 0.00 & 0.00  & 0.00 & 0.00 & 0.00 \\
    \cline{2-8}
    & \multirow{4}{*}{10} & Green & 0.99 & 0.99 & 1.00 & 0.99 & 0.98 \\
    & & Amber & 0.01 & 0.01 & 0.00 & 0.01 & 0.02\\
    & & Red & 0.00 & 0.00 & 0.00 & 0.00 & 0.00 \\
    \cline{2-8}
    & \multirow{4}{*}{15} & Green & 1.00 & 1.00 & 1.00 & 1.00 & 1.00 \\
    & & Amber & 0.00 & 0.00 & 0.00 & 0.00 & 0.00\\
    & & Red & 0.00 & 0.00 & 0.00 & 0.00 & 0.00\\
    \cline{1-8}
    \multirow{9}{*}{99 \%} & \multirow{4}{*}{5} & Green & 0.78  & 0.75  & 0.81 & 0.79 & 0.50\\
    & & Amber & 0.22 & 0.25  & 0.19 & 0.21 & 0.50 \\
    & & Red & 0.00 & 0.00  & 0.00 & 0.00  & 0.00 \\
    \cline{2-8}
    & \multirow{4}{*}{10} & Green & 0.74 & 0.72 & 0.76 & 0.75 & 0.47\\
    & & Amber & 0.26 & 0.28 & 0.24 & 0.25 & 0.53 \\
    & & Red & 0.00 & 0.00 & 0.00 & 0.00 & 0.00\\
    \cline{2-8}
    & \multirow{4}{*}{15} & Green & 0.86 & 0.84 & 0.77 & 0.70 & 0.55\\
    & & Amber & 0.14 & 0.16 & 0.23 & 0.30 & 0.45 \\
    & & Red & 0.00 & 0.00 & 0.00 & 0.00 & 0.00\\
    \cline{1-8}
\end{tabular}}{}
\end{mdframed}
\end{table}
\end{singlespace}

Figure \ref{fig:VaR_returns_2019} further illustrates how the models behave by showing how the returns and VaR estimates evolve over time in 2019 for a single portfolio of size 10. Figure \ref{fig:VaR_returns_2019} confirms the similarity between the different methods in 2019. However, although all of the VaR estimates are quite similar, the methods based on the volatility sensitive conjugate prior and the DCC-GARCH(1,1) give rise to more time varying estimates. 

\begin{figure}[H]
\begin{mdframed}[backgroundcolor=black!10,rightline=false,leftline=false]
    \centering
    \includegraphics[width=0.85\linewidth]{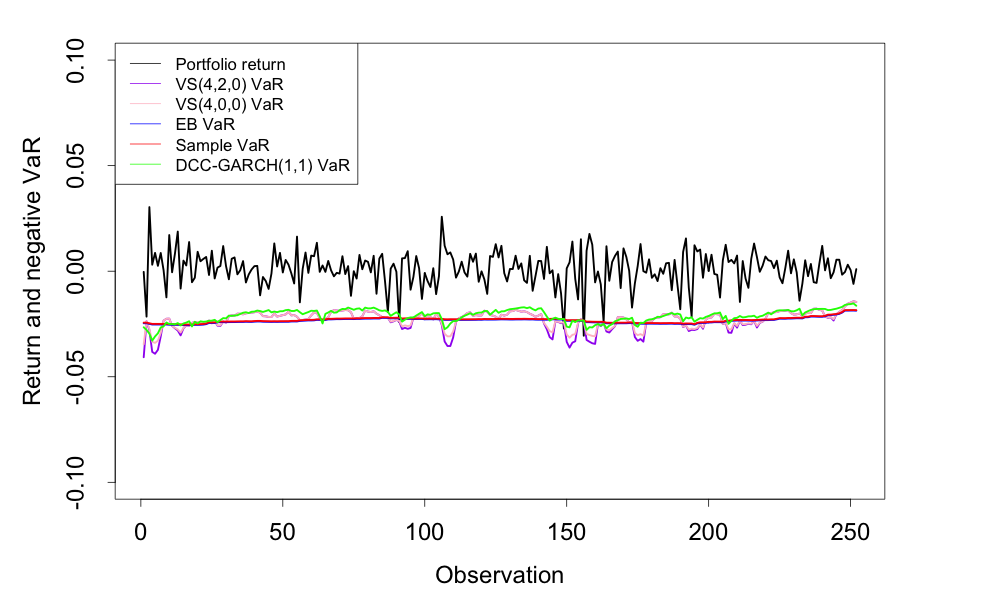}
    \caption{Returns and negative VaR estimates at the 99 \% level for one portfolio of size 10 in 2019.}
    \label{fig:VaR_returns_2019}
 \end{mdframed}
\end{figure}

Table \ref{tab:empirical_2020} shows how the different models compare during a more turbulent year as in 2020. We observe a clear advantage of the methods based on the volatility sensitive conjugate prior, particularly VS(4,2,0). At the 97.5 \% level, almost all of the test results are classified as 'Green' when using this method whereas the other methods result in a majority of the results being classified as 'Amber' or 'Red'. At the 99 \% level we see that the new method gives rise to some 'Amber' results, but still none of the results are classified as 'Red'. The other methods have basically no results in the 'Green' zone at the 99 \% level. Even the DCC-GARCH(1,1) method, which is also a heteroscedastic model, fails to estimate VaR during very turbulent periods. 

\begin{singlespace}
\begin{table}[H]
\footnotesize
\begin{mdframed}[backgroundcolor=black!10,rightline=false,leftline=false]
\centering
\caption{Proportion of VaR estimations classified as 'Green', 'Amber' and 'Red' in the Basel backtest using different portfolio sizes and VaR estimation methods in 2020. }
\label{tab:empirical_2020}
{\begin{tabular}{| c | c | c | c c c c c |}
  \hline
  \makecell{VaR \\level} & \makecell{Portfolio\\ size} & \makecell{Basel \\ zone} & \makecell{VS\\(4,2,0)} & \makecell{VS\\(4,0,0)} & EB & Sample & \makecell{DCC-GARCH\\(1,1)}\\
  \hhline{|=|=|=|=====|}
    \multirow{9}{*}{97.5 \%} & \multirow{4}{*}{5} & Green & 1.00  & 0.85  & 0.00 & 0.00 & 0.12 \\
    & & Amber & 0.00 & 0.15  & 0.76 & 0.74 & 0.88 \\
    & & Red & 0.00 & 0.00  & 0.24 & 0.26 & 0.00 \\
    \cline{2-8}
    & \multirow{4}{*}{10} & Green & 1.00 & 0.88 & 0.00 & 0.00 & 0.09 \\
    & & Amber & 0.00 & 0.12 & 0.82 & 0.80 & 0.89\\
    & & Red & 0.00 & 0.00 & 0.18 & 0.20 & 0.02\\
    \cline{2-8}
    & \multirow{4}{*}{15} & Green & 1.00 & 0.91 & 0.00 & 0.00 & 0.03 \\
    & & Amber & 0.00 & 0.09 & 0.84 & 0.77 & 0.94\\
    & & Red & 0.00 & 0.00 & 0.16 & 0.23 & 0.03\\
    \cline{1-8}
    \multirow{9}{*}{99 \%} & \multirow{4}{*}{5} & Green & 0.63  & 0.13  & 0.00 & 0.00 & 0.01\\
    & & Amber & 0.37 & 0.85  & 0.02 & 0.02 & 0.69 \\
    & & Red & 0.00 & 0.02  & 0.98 & 0.98  & 0.30 \\
    \cline{2-8}
    & \multirow{4}{*}{10} & Green & 0.58 & 0.06 & 0.00 & 0.00 & 0.00\\
    & & Amber & 0.42 & 0.94 & 0.00 & 0.00 & 0.66 \\
    & & Red & 0.00 & 0.00 & 1.00 & 1.00 & 0.34\\
    \cline{2-8}
    & \multirow{4}{*}{15} & Green & 0.47 & 0.00 & 0.00 & 0.00 & 0.00\\
    & & Amber & 0.53 & 1.00 & 0.00 & 0.00 & 0.43 \\
    & & Red & 0.00 & 0.00 & 1.00 & 1.00 & 0.57\\
    \cline{1-8}
\end{tabular}}{}
\end{mdframed}
\end{table}
\end{singlespace}

Figure \ref{fig:VaR_returns_2020} provides a deeper insight into how the different models behave by showing how the returns and VaR estimates evolve over time in 2020 for a single portfolio of size 10. It shows that the new method can quickly adapt to changing market conditions as in March 2020. The homoscedastic models are not at all suitable in such circumstances, indicated by their very moderate increase in VaR. The DCC-GARCH(1,1) model seems to scale up VaR during the turbulent period but evidently not enough. However, it is interesting to note that the VaR estimates based on the new method and the DCC-GARCH(1,1) method behave very synchronously. 

\begin{figure}[H]
\begin{mdframed}[backgroundcolor=black!10,rightline=false,leftline=false]
    \centering
    \includegraphics[width=0.85\linewidth]{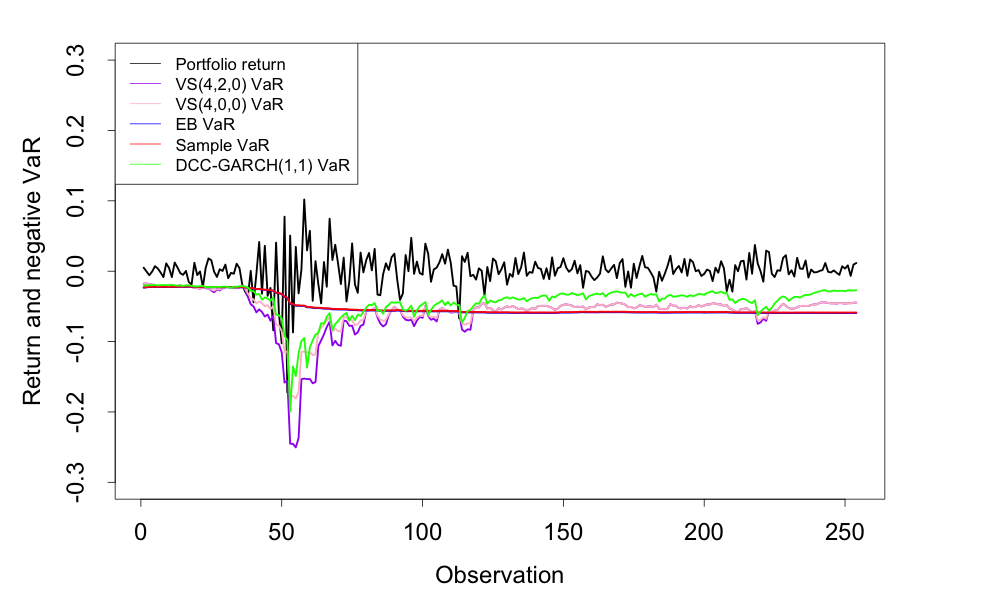}
    \caption{Returns and negative VaR estimates at the 99 \% level for one portfolio of size 10 in 2020.}
    \label{fig:VaR_returns_2020}
 \end{mdframed}
\end{figure}

To sum up, during a stable year such as 2019 all of the considered methods perform quite well in terms of VaR estimation. Due to the lack of rapid changes in the market conditions, methods that account for heteroscedacity are not advantageous. However, during turbulent years such as 2020, we see a clear advantage of heteroscedastic models. In terms of VaR estimation, the new method based on the volatility sensitive conjugate prior gives the best results.

\section{Conclusion}\label{discussion}
When working with financial data it is important to keep in mind that there exists no perfect model. A model can only be considered 'better' or 'worse' based on historical data and some statistical measures. Many of the standard statistical tests would reject several of the models used by practitioners, but that does not necessarily mean that the models that they use are useless. In reality, usefulness is a tradeoff between accuracy, simplicity and speed.

In this paper, we have suggested a new way to specify the hyperparameters in the conjugate prior which makes it possible to capture volatility clustering in financial data. The new method is simple and very fast. Moreover, it solves the problem of certainty specification by automatically setting the degree of belief depending on how risk cautious the investor is. It has been illustrated using both simulated and real market data that the new method estimates VaR very accurately when using the Basel backtest. Since CVaR is backtested in terms of VaR, this also qualifies the model for providing good CVaR estimates according to Basel.

Although the new method shows promising results in terms of VaR estimation, it is likely possible to improve it by considering dynamic correlations in the specification of the prior. This would lead to an even more flexible model.   

\section*{Acknowledgment}
This research was partly supported by the Swedish Research Council (VR) via the project ``Bayesian Analysis of Optimal Portfolios and Their Risk Measures''.

\bibliography{references}

\begin{thebibliography}{}

\bibitem[Alexander and Baptista, 2002]{alexander2002economic}
Alexander, G.~J. and Baptista, A.~M. (2002).
\newblock Economic implications of using a mean-{VaR} model for portfolio
  selection: {A} comparison with mean-variance analysis.
\newblock {\em Journal of Economic Dynamics and Control}, 26(7):1159--1193.

\bibitem[Barry, 1974]{Barry1974}
Barry, C.~B. (1974).
\newblock {Portfolio analysis under uncertain means, variances, and
  covariances}.
\newblock {\em Journal of Finance}, 29:515--522.

\bibitem[Bauder et~al., 2020]{bauder2020bayesian}
Bauder, D., Bodnar, T., Parolya, N., and Schmid, W. (2020).
\newblock Bayesian inference of the multi-period optimal portfolio for an
  exponential utility.
\newblock {\em Journal of Multivariate Analysis}, 175.
\newblock 104544.

\bibitem[Bauder et~al., 2021]{bauder2018bayesian}
Bauder, D., Bodnar, T., Parolya, N., and Schmid, W. (2021).
\newblock Bayesian mean--variance analysis: {O}ptimal portfolio selection under
  parameter uncertainty.
\newblock {\em Quantitative Finance}, 21:221--242.

\bibitem[Bernardo and Smith, 2000]{bernardo2000bayesian}
Bernardo, J.~M. and Smith, A.~F. (2000).
\newblock {\em Bayesian Theory}, volume 405.
\newblock John Wiley \& Sons.

\bibitem[Black and Litterman, 1992]{black1992global}
Black, F. and Litterman, R. (1992).
\newblock Global portfolio optimization.
\newblock {\em Financial Analysts Journal}, 48:28--43.

\bibitem[Bodnar et~al., 2022]{BodnarLindholmNiklassonThorsen2020}
Bodnar, T., Lindholm, M., Niklasson, V., and Thors{\'e}n, E. (2022).
\newblock Bayesian portfolio selection using var and cvar.
\newblock {\em Applied Mathematics and Computation}, 427:127120.

\bibitem[Christoffersen, 1998]{christoffersen1998evaluating}
Christoffersen, P.~F. (1998).
\newblock Evaluating interval forecasts.
\newblock {\em International Economic Review}, pages 841--862.

\bibitem[Ding and Granger, 1996]{ding1996modeling}
Ding, Z. and Granger, C.~W. (1996).
\newblock Modeling volatility persistence of speculative returns: {A} new
  approach.
\newblock {\em Journal of Econometrics}, 73(1):185--215.

\bibitem[Ding et~al., 1993]{ding1993long}
Ding, Z., Granger, C.~W., and Engle, R.~F. (1993).
\newblock A long memory property of stock market returns and a new model.
\newblock {\em Journal of Empirical Finance}, 1(1):83--106.

\bibitem[Du and Escanciano, 2017]{du2017backtesting}
Du, Z. and Escanciano, J.~C. (2017).
\newblock Backtesting expected shortfall: {A}ccounting for tail risk.
\newblock {\em Management Science}, 63(4):940--958.

\bibitem[Engle, 2002]{engle2002dynamic}
Engle, R. (2002).
\newblock Dynamic conditional correlation: {A} simple class of multivariate
  generalized autoregressive conditional heteroskedasticity models.
\newblock {\em Journal of Business \& Economic Statistics}, 20(3):339--350.

\bibitem[Engle, 1982]{engle1982autoregressive}
Engle, R.~F. (1982).
\newblock Autoregressive conditional heteroscedasticity with estimates of the
  variance of {U}nited {K}ingdom inflation.
\newblock {\em Econometrica}, pages 987--1007.

\bibitem[Escanciano and Pei, 2012]{escanciano2012pitfalls}
Escanciano, J.~C. and Pei, P. (2012).
\newblock Pitfalls in backtesting historical simulation {VaR} models.
\newblock {\em Journal of Banking \& Finance}, 36(8):2233--2244.

\bibitem[Frost and Savarino, 1986]{frost1986empirical}
Frost, P.~A. and Savarino, J.~E. (1986).
\newblock An empirical {B}ayes approach to efficient portfolio selection.
\newblock {\em Journal of Financial and Quantitative Analysis}, 21(3):293--305.

\bibitem[Kerkhof and Melenberg, 2004]{kerkhof2004backtesting}
Kerkhof, J. and Melenberg, B. (2004).
\newblock Backtesting for risk-based regulatory capital.
\newblock {\em Journal of Banking \& Finance}, 28(8):1845--1865.

\bibitem[Klein and Bawa, 1976]{KleinBawa1976}
Klein, R. and Bawa, V. (1976).
\newblock {The effect of estimation risk on optimal portfolio choice}.
\newblock {\em Journal of Financial Economics}, 3:215--231.

\bibitem[Kolm and Ritter, 2017]{kolm2017bayesian}
Kolm, P. and Ritter, G. (2017).
\newblock On the {B}ayesian interpretation of {B}lack--{L}itterman.
\newblock {\em European Journal of Operational Research}, 258(2):564--572.

\bibitem[Kratz et~al., 2018]{kratz2018multinomial}
Kratz, M., Lok, Y.~H., and McNeil, A.~J. (2018).
\newblock Multinomial {VaR} backtests: {A} simple implicit approach to
  backtesting expected shortfall.
\newblock {\em Journal of Banking \& Finance}, 88:393--407.

\bibitem[Lee et~al., 2006]{lee2006study}
Lee, M.-C., Chiou, J.-S., and Lin, C.-M. (2006).
\newblock A study of value-at-risk on portfolio in stock return using {DCC}
  multivariate {GARCH}.
\newblock {\em Applied Financial Economics Letters}, 2(3):183--188.

\bibitem[P{\'a}stor, 2000]{pastor2000portfolio}
P{\'a}stor, L. (2000).
\newblock Portfolio selection and asset pricing models.
\newblock {\em The Journal of Finance}, 55(1):179--223.

\bibitem[P{\'a}stor and Stambaugh, 2000]{pastor2000comparing}
P{\'a}stor, L. and Stambaugh, R.~F. (2000).
\newblock Comparing asset pricing models: {An} investment perspective.
\newblock {\em Journal of Financial Economics}, 56(3):335--381.

\bibitem[Santos et~al., 2013]{santos2013comparing}
Santos, A.~A., Nogales, F.~J., and Ruiz, E. (2013).
\newblock Comparing univariate and multivariate models to forecast portfolio
  value-at-risk.
\newblock {\em Journal of Financial Econometrics}, 11(2):400--441.

\bibitem[{The Basel Committee on Banking Supervision}, 2019]{basel4}
{The Basel Committee on Banking Supervision} (2019).
\newblock {Minimum capital requirements for market risk}.
\newblock Online: accessed 12-April-2021.

\bibitem[Tu and Zhou, 2010]{tu2010incorporating}
Tu, J. and Zhou, G. (2010).
\newblock Incorporating economic objectives into {B}ayesian priors: {P}ortfolio
  choice under parameter uncertainty.
\newblock {\em Journal of Financial and Quantitative Analysis}, pages 959--986.

\bibitem[Winkler, 1973]{winkler1973bayesian}
Winkler, R.~L. (1973).
\newblock Bayesian models for forecasting future security prices.
\newblock {\em Journal of Financial and Quantitative Analysis}, pages 387--405.

\bibitem[Ziggel et~al., 2014]{ziggel2014new}
Ziggel, D., Berens, T., Wei{\ss}, G.~N., and Wied, D. (2014).
\newblock A new set of improved value-at-risk backtests.
\newblock {\em Journal of Banking \& Finance}, 48:29--41.

\end{thebibliography}

\end{document}